\documentclass[11pt]{article}
\usepackage{jheppub}

\usepackage{color}
\usepackage{amsmath}

\usepackage{verbatim}   
\usepackage{subfigure}  
\usepackage{amsfonts}
\usepackage{amssymb}
\usepackage{mathrsfs}
\usepackage{graphicx}
\usepackage{slashed}
\usepackage{amsthm}

\usepackage{graphicx}
\usepackage{dcolumn}
\usepackage{bm}

\usepackage{color}
\usepackage{amsmath}

\usepackage{verbatim}   
\usepackage{subfigure}  
\usepackage{amsfonts}
\usepackage{amssymb}
\usepackage{mathrsfs}
\usepackage{graphicx}
\usepackage{slashed}
\usepackage{amsthm}
\newcommand{\kahler}{K\"{a}hler }

\newcommand{\C}{\mathbb{C}}     
    
\newcommand{\dd}{\mathrm{d}}   
\def\beq{\begin{equation}\begin{aligned}}
\def\eeq{\end{aligned}\end{equation}}

\def\OO{\mathcal{O}}

\def\d{\textrm{d}}
\def\d{\textrm{d}}
\def\tr{\textrm{tr}}
\def\OO{\mathcal{O}}
\def\a{\alpha'}

\def\bar#1{\overline{#1}}

\def\inv{^{\raise.15ex\hbox{${\scriptscriptstyle -}$}\kern-.05em 1}}
\def\lbar{{\lower.35ex\hbox{$\mathchar'26$}\mkern-10mu\lambda}} 

\def\p{\partial}
\def\bp{\bar\partial}
\def\F{\mathcal{F}}
\def\R{\mathcal{R}}
\def\End{{\rm End}}
\def\Q{\mathcal{Q}}
\def\M{\mathcal{M}}
\def\H{\mathcal{H}}
\def\C{\mathcal{C}}
\def\g{\bold g}

\let\p=\partial
\let\<=\langle
\let\>=\rangle

\let\+=\uparrow

\def\H{\mathcal{H}}

\definecolor{mygreen}{RGB}{29,145,47}
\definecolor{mypurple}{RGB}{164,64,214}
\definecolor{myorange}{RGB}{199,146,32}

\newtheorem{Theorem}{Theorem}
\newtheorem*{Theorem*}{Theorem}

\definecolor{orange}{rgb}{1,0.5,0}

\begin{document}
\title{The Heterotic Superpotential and Moduli}


\author[a]{Xenia de la Ossa,}
\author[b]{Edward Hardy}
\author[c,d,e]{and Eirik Eik Svanes}

\affiliation[a]{Mathematical Institute, University of Oxford, Andrew Wiles Building, Radcliffe Observatory Quarter, Woodstock Road, Oxford OX2 6GG, UK}
\affiliation[b]{Abdus Salam International Centre for Theoretical Physics, Strada Costiera 11, 34151, Trieste, Italy}
\affiliation[c]{Sorbonne Universit\'es, UPMC Univ Paris 06, UMR 7589, LPTHE, F-75005, Paris, France}
\affiliation[d]{CNRS, UMR 7589, LPTHE, F-75005, Paris, France}
\affiliation[e]{Sorbonne Universit\'es, Institut Lagrange de Paris, 98 bis Bd Arago, 75014 Paris, France\\}

\emailAdd{delaossa@maths.ox.ac.uk}
\emailAdd{ehardy@ictp.it} 
\emailAdd{esvanes@lpthe.jussieu.fr}


\abstract{We study the four-dimensional effective theory arising from ten-dimensional heterotic supergravity compactified on manifolds with torsion. In particular, given the heterotic superpotential appropriately corrected at $\OO(\a)$ to account for the Green-Schwarz anomaly cancellation mechanism, we investigate properties of four-dimensional Minkowski vacua of this theory. Considering the restrictions arising from F-terms and D-terms we identify the infinitesimal massless moduli space of the theory. We show that it agrees with the results that have recently been obtained from a ten-dimensional perspective where supersymmetric Minkowski solutions including the Bianchi identity  correspond to an integrable holomorphic structure, with infinitesimal moduli calculated by its first cohomology. As has recently been noted, interplay of complex structure and bundle deformations through holomorphic and anomaly constraints can lead to fewer moduli than may have been expected. We derive a relation between the number 
of complex structure and bundle moduli removed from the low energy theory in this way, and give conditions for there to be no complex structure moduli or bundle moduli remaining in the low energy theory. The link between Yukawa couplings and obstruction theory is also briefly discussed.}

\maketitle


\section{Introduction}
While the low energy effective theories arising from compactifications of heterotic string theory are attractive for their ease of reproducing the matter content and gauge interactions of the Standard Model, moduli stabalisation has long remained a significant challenge. For mathematical clarity and the ability to calculate many explicit examples, much of the literature so far has focused on the Calabi-Yau compactifications first considered in \cite{Candelas:1985en}. However, there is no intrinsic reason to prefer these to compacifications including torsion. The latter may even, speculatively, be hoped to be more amiable to moduli stabilisation. Indeed it has been argued that torsion can be used in a similar fashion to flux to stabilise moduli in heterotic compactifications \cite{Gurrieri:2004dt, Gurrieri:2007jg, Klaput:2012vv}.

In this paper we explore the four-dimensional effective theories arising from compacifications with torsion, focusing especially on identifying their moduli space. Of course, a large body of related work already exists. The general supersymmetry conditions and structures for such theories were given in the papers \cite{Strominger1986253,Hull:1985zy,Hull:1986kz,deWit:1986xg}, and the ten-dimensional action was obtained in \cite{Bergshoeff:1988nn,Bergshoeff:1989de} while an early study of the relevant geometrical tools appeared in \cite{Lust:1986ix}. Further study, explicit examples of flux compactifications, and important results regarding the moduli of these theories followed for example in \cite{Becker:2003yv,Becker:2003sh}. Particularly relevant to our current work, \cite{Becker:2005nb, Becker:2006xp} considered the moduli space of heterotic compactifications under certain simplifying assumptions. Additionally, in \cite{Becker:2003gq} the relation of the superpotential to the moduli space was studied, and it 
was shown that some torsional solutions typically lead to a small compactification radius: a regime in which it can be hard to make sense of the supergravity approximation. 

In the torsion free Calabi-Yau case, the infinitesimal moduli corresponding to simultaneous deformations of the bundle together with the complex structure of the Calabi-Yau, have been identified and expressed in terms of Atiyah Classes in \cite{Anderson:2011ty}. In particular, the moduli space is not simply the direct product of the complex structure and bundle moduli. Instead, as a deformation of the complex structure alters the definition of a holomorphic bundle, if a compensating deformation in the bundle is not possible in this direction it will not correspond to a massless field. From the perspective of the 4d theory, this was understood as F-terms associated to (possibly heavy) modes lifting the potential in some directions, and examples have been obtained where no complex structure moduli remain in the low energy theory \cite{Anderson:2010mh, Anderson:2011cza}. 

Recently the infinitesimal moduli of solutions of the Strominger system with torsion from the perspective of the ten-dimensional supergravity theory has been understood \cite{delaOssa:2014cia, Anderson:2014xha, GarciaFernandez:2015hja}. In particular, it was shown how to  include the heterotic anomaly in a natural extension of the holomorphic Atiyah story. Specifically, in \cite{delaOssa:2014cia} it was shown how the Strominger system can be rephrased as an integrability condition for a holomorphic structure on a double extension bundle. This then includes the complex base, the holomorphic bundle {\it and the heterotic Bianchi identity}. It is the purpose of the present paper to re-derive the results of this paper from the perspective of the four-dimensional superpotential, thus generalising the work done in \cite{Anderson:2010mh}, which focused on the complex structure moduli of the base together with the bundle moduli. It should be noted that a couple of key features of the Strominger system are not accounted for by the holomorphic double 
extension. In particular, the manifold is required to be conformally balanced while the bundles should satisfy the Yang-Mills condition. We shall see how these conditions are accounted for in the moduli problem, and how they relate to D-term conditions in the four-dimensional effective supergravity. 

The paper is organised as follows.  In Section \ref{sec:sptot}, we study the first order variations of the superpotential with torsion which produces some of the well known constraints on the compactifying 6 dimensional manifold $X$ and the gauge vector bundle $V$ over $X$, namely that $X$ is a complex manifold with a nowhere vanishing holomorphic three form $\Omega$, that $\End(V)$ and $\End(TX)$ are equipped with holomorphic connections and moreover, that the three form flux $H$ is must be related to the hermitian form $\omega$ by the formula $H = \dd^c\omega$ (see Appendix \ref{app:SU3}). 
In Section \ref{sec:ms} we show how the infinitesimal moduli in torsional theories can be derived by considering the four-dimensional superpotential of the low energy effective theory. The result is found to agree with the ten-dimensional perspective apart from an extra restriction on the moduli space. We also review the computation of the moduli space in terms of maps on cohomologies, and investigate these maps further in relation to the lifting of some of the bundle and complex structure moduli. In particular, one of the main results of this paper uses these maps to prove that at least as many complex structure moduli as bundle moduli are lifted from the four-dimensional massless spectrum. 

In Section \ref{sec:Dterm} we show that the additional restriction mentioned above which is obtained in the ten-dimensional supergravity computation, not arising from F-terms in the effective theory, is generated by four-dimensional D-terms completing the correspondence between the work in \cite{delaOssa:2014cia} and the study of the superpotential of the effective four-dimensional theory. Similarly to the Calabi-Yau compactifications which had been studied previously, these restrictions come from the requirement that the variation of the ten-dimensional gaugino fields vanishes.  This is equivalent to the polystability of the gauge bundle. In Section \ref{sec:EndTconn} we also discuss deformations of the $\End(TX)$-valued connection appearing in the Bianchi identity. As shown in \cite{delaOssa:2014msa} these correspond to field redefinitions and can hence be ignored from a physics point of view. They are however needed for the most natural implementation of the mathematical structure presented here. We also mention the role of higher order obstructions and Yukawa couplings in Section \ref{sec:yuk}, although many details of this topic are 
left to further work. 

Finally, in Section \ref{sec:examples} we consider example compactifications. We revisit the standard embedding \cite{Candelas:1985en}, showing that as expected no margial moduli will be lifted in this case. We then turn to consider conditions for  all bundle and/or complex structure moduli to be lifted (previous work on this topic has appeared in \cite{Anderson:2010mh, Anderson:2011cza, Anderson:2011ty}). The conclusions we draw are not specific to Calabi-Yau compactifications, extending previous results to more general torsional compactifications. The most phenomenologically interesting examples of smooth compactifications that our results are relevant to  (so far in the literature) are of the form where a large volume compact $\a\rightarrow0$ Calabi-Yau limit exists.

We have left some technicalities concerning $SU(3)$-structures and the $\a$-expansion to Appendix \ref{app:SU3} and \ref{app:exp}.

\section{The Heterotic Superpotential with Torsion}\label{sec:sptot}
We consider compactifications to four dimensions which preserves $N=1$ supersymmetry. This implies that the six-dimensional internal space has a nowhere vanishing spinor, which in turn implies that the internal space has an $SU(3)$-structure. The heterotic superpotential for a compactification on a manifold with an $SU(3)$-structure $(X,\Omega,\omega)$ with vector bundle $(V,A)$ (details of our notation are defined in Appendix \ref{app:SU3}) is \cite{Becker:2003yv, Cardoso:2003af, Becker:2003gq, Gurrieri:2004dt}
\begin{equation}
\label{eq:suppot}
W=\int_X(H+i\d\omega)\wedge\Omega\:,
\end{equation}
where 
\begin{equation}
\label{eq:anomaly}
H=H_0+\frac{\a}{4}(\omega_{CS}^A-\omega_{CS}^\nabla)\:,
\end{equation}
with $H_0=\d B$ for some local two-form potential $B$, and the Chern-Simons forms are
\begin{equation*}
\omega_{CS}^A=\tr\:\Big(A\wedge\d A+\frac{2}{3}\,A^3\Big)\:,
\end{equation*}
and similar for $\omega_{CS}^\nabla$. Here $\nabla$ is some spin-connection on the tangent bundle of $X$.\footnote{$\nabla$ is usually taken to be an instanton connection \cite{Ivanov:2009rh, Martelli:2010jx},
\begin{equation*}
R_{mn}\gamma^{mn}\eta=0\:,
\end{equation*}
where $\eta$ is the six-dimensional spinor on $X$ and $R$ is the curvature of $\nabla$. We will not start with this assumption, but rather derive it from the superpotential and D-term potential.}

\subsection{First Order Deformations}
The supersymmetry conditions may be derived by considering the first order deformation of $W$. This gives the F-terms of the fields in the action of the low energy effective theory and therefore they must vanish. This has been previously discussed in, for example, \cite{Becker:2003yv, Cardoso:2003af, Becker:2005nb}, and the conditions derived are equivalent to the 10d supersymmetry conditions in \cite{Strominger1986253}. 

Considering a first order deformation of $W$, Eq.\eqref{eq:suppot}, the supersymmetry condition is
\begin{equation}
\label{eq:delta1W}
\delta W=\int_X\left(\frac{\a}{2}\big(\tr\:(\delta A\wedge F)-\tr\:(\delta\Theta\wedge R)\big)+\d\tau\right)\wedge\Omega+\int_X(H+i\d\omega)\wedge\delta\Omega = 0\:,
\end{equation}
where $F$ and $R$ are the curvatures of $A$ and $\nabla$ respectively, $\Theta$ is the connection one-form of $\nabla$, and where we have defined
\begin{equation}
\label{eq:complKahler}
\tau=\delta B+i\delta\omega+\frac{\a}{4}\big(\tr\:(\delta A\wedge A)-\tr\:(\delta\Theta\wedge\Theta)\big)\:.
\end{equation}
We may think of $\tau$ as a deformation of  the complexified hermitian form, appropriately adjusted by the connection terms at $\OO(\a)$. Using the Green-Scwarz mechanism, $\tau$ can be shown to be gauge invariant \cite{delaOssa}. Moreover, due to flux quantisation $\tau$ is a globally well-defined two-form.
 
The first two terms of equation \eqref{eq:delta1W} then tell us that supersymmetry requires
\begin{equation*}
F\wedge\Omega=0\:,\;\;\;R\wedge\Omega=0\:,
\end{equation*}
or equivalently $F^{(0,2)}=R^{(0,2)}=0$ with respect to the (almost) complex structure determined by $\Omega$. Hence, $(V,A)$ and $(TX,\nabla)$ are holomorphic vector bundles. Integrating by parts the third term in Eq.\eqref{eq:delta1W}, we get
\begin{equation*}
\d\Omega=0\:,
\end{equation*}
that is, $\Omega$ is a holomorphic three-form. It follows that the almost complex structure determined by $\Omega$ is integrable. Finally, for the last term of Eq.\eqref{eq:delta1W}, we recall that 
\[\delta\Omega=K\,\Omega+\chi~\]
 where $\chi\in \Omega^{(2,1)}(X)$. Therefore, 
\begin{equation*}
(H+i\d\omega)^{(0,3)}=(H+i\d\omega)^{(1,2)}=0\:.
\end{equation*}
It follows that
\begin{equation}
H=i(\p-\bp)\omega=\d^c\omega\:.\label{eq:H}
\end{equation}

In conclusion,  we have found that the supersymmetry condition
\begin{equation*}
\delta W=0
\end{equation*}
requires that $(X,\Omega,\omega)$ to be a complex manifold which admits a no-where vanishing holomorphic three-form $\Omega$, and that the bundles $(V,A)$ and $(TX, \nabla)$ are holomorphic vector bundles.  Moreover, the three form flux $H$ is required to be given by equation \eqref{eq:H}. Note also that these conditions imply that $W=0$, implying that the vacuum is Minkowski. It is worth remarking at this point that the remaining constraints in the Strominger system  (see Appendix \ref{app:SU3}), namely the conformally balanced condition and the Yang-Mills equations for the connections on  $\End(V)$ and $\End(TX)$, are obtained from the vanishing of $D$-term superpotential as we will see later in Section \ref{sec:YM} and \ref{sec:EndTconn}.  

\section{The Moduli Space} \label{sec:ms}

\subsection{Second Order Deformations}
We now consider the massless moduli degrees of freedom associated the superpotential Eq.\eqref{eq:suppot}. The Hesse matrix can be found in e.g. \cite{Denef:2004cf}, and at an extremum of the four-dimensional F-term potential it reads 
\begin{align*}
V_{I\bar J}=e^{{\cal K}}\left(-R_{I\bar J K\bar L}\, \bar F^JF^{\bar K}+{\cal K}_{I\bar J} \vert F \vert^2 -F_I\bar F_{\bar J}+(D_IF_K)(\bar D_J\bar F^K) -2{\cal K}_{I\bar J}\vert W\vert^2\right) .
\end{align*}
Here ${I,J,..}$ denote generic moduli directions, $\cal K$ is the K\"ahler potential for the K\"ahler metric on the heterotic moduli space, $D_I$ are moduli space covariant derivatives, $F_I=D_IW$ are the F-terms, and $R_{I\bar J K\bar L}$ is the curvature of the K\"ahler metric ${\cal K}_{I\bar J}$. This metric can be computed \cite{delaOssa}, and is a rather complicated function of the moduli, and whose full expression will not be required for our purposes. We are only interested in an understanding of how some of the moduli become massive, but not on the explicit values of these masses. Concerning ourselves with supersymmetric Minkowski vacua, both the superpotential and the F-terms vanish ($W=\delta W=0$), and the mass-matrix reduces to
\begin{equation}
\label{eq:4dmass}
V_{I\bar J}=e^{\cal K}\, (\p_I\p_K W)(\p_{\bar J}\p_{\bar L}\bar W)\,{\cal K}^{K\bar L}\:.
\end{equation}
Working in a basis for moduli where the mass matrix is diagonal, it is easy to see that the modulus corresponding to the direction $I$ is massless if and only if $\p_I\p_KW=0$ for all $K$.

To find the moduli satisfying this condition we need to consider not only deformations that preserve supersymmetry, and so have a chance of being massless, but also deformations breaking  supersymmetry, which although heavy can generate masses for some of the light moduli. Indeed, if $\delta_2\delta_1W$ is non-zero, where $\delta_2$ preserves the supersymmetry conditions and $\delta_1$ does not, then this will lead  to a large mass for the naively supersymmetric modulus $\delta_2$. Physically, $\delta_2 \delta_1 W = \delta_2 \left(F_1\right)$, where $F_1$ is the F-term associated to $\delta_1$, so a potential is generated in the direction $\delta_2$ if this does not vanish.\footnote{As discussed in  \cite{Anderson:2010mh} such heavy modes are never explicitly present in the 4d effective theory. We regard the agreement with rigorous 10d computations as a confirmation that this approach is consistent.} Of course, we will not worry about deformations where both $\delta_i$ are non-supersymmetric, 
as these are lifted out from the theory from the start. As discussed in the Introduction and \cite{Anderson:2010mh}, the result of these considerations is that the true moduli of the theory are a combination of complex structure, K\"ahler, and bundle deformations, which can be given in terms of a subspace of the original cohomologies. We will see how this works for the Strominger system including the anomaly cancellation condition in this section.

Let us begin by performing a second order deformation of $W$ at the supersymmetric locus $W=\delta W=0$. We take $\delta_1$ to be a generic deformation while $\delta_2$ is massless deformation. According to the above discussion and \eqref{eq:4dmass}, we need that
\begin{equation}
\begin{split}
\delta_2\delta_1W\vert_0&=\int_X\d\tau_1\wedge\delta_2\Omega
+\int_X\frac{\a}{2}\left(\tr\left(\delta_1A\wedge\delta_2(F\wedge\Omega)\right)-\tr\left(\delta_1\Theta\wedge\delta_2(R\wedge\Omega)\right)\right)\\[5pt]
&\quad+\int_X\delta_2(H+i\d\omega)\wedge\delta_1\Omega+\int_X(H+i\d\omega)\wedge\delta_2\delta_1\Omega= 0\:,
\end{split}\label{eq:delta2W}
\end{equation}
for all deformations $\delta_1$ if the deformation $\delta_2$ is to be massless. Here the zero in $\delta_2\delta_1W\vert_0$ denotes we are imposing the ten-dimensional supersymmetry conditions found in the previous section. 
 
\subsection{Holomorphic Bundle Deformations}
Consider the first term of Eq.\eqref{eq:delta2W}.  For generic $\tau_1$, this gives the condition
\begin{equation}
\label{eq:presInt}
\d\delta_2\Omega=0\:,
\end{equation}
where, as before, $\delta_2\Omega=K_2\,\Omega+\chi_2$, for $\chi_2\in \Omega^{(2,1)}(X)$, and $K_2$ can be taken to be a constant by modding out by diffeomorphisms of $X$. A nontrivial deformation of $\Omega$ then corresponds to an element
\begin{equation*}
\chi_2\in H^{(2,1)}(X)\cap\ker(\d)\:.
\end{equation*}
Recall that the existence of a well defined no-where vanishing $(3,0)$ form $\Omega$ defines an isomorphism between the cohomologies
\begin{equation*}
H^{(2,1)}(X)\cong H^{(0,1)}(TX)\:,
\end{equation*}
where $TX$ denotes the holomorphic tangent bundle.
The elements $\Delta_2\in H^{(0,1)}(TX)$ correspond to a non-trivial deformation of the complex structure $J$ which preserve integrability, however they do not necessarily preserve the holomorphicity of $\Omega$.  Hence the allowed deformations of the complex structure $\chi_2$, are those for which there exists a $\d$-closed representative in its  the cohomology class in $H^{(2,1)}(X)$. This holds true for all cohomology classes whenever the $\p\bar\p$-lemma is satisfied. To see this, consider
\begin{equation*}
[\chi]\in H^{(2,1)}(X)\;\;\Rightarrow\;\;\bp\p\chi=0\;\;\Rightarrow\;\;\p\chi=\p\bp\beta
\end{equation*}
for some $(2,0)$-form $\beta$. The last implication follows from the $\p\bp$-lemma. It follows that $\p(\chi-\bp\beta)=0$, and $[\chi-\bp\beta]=[\chi]\in H^{(2,1)}(X)$. In particular, this will be true whenever there exists a large volume $\a\rightarrow0$ compact Calabi-Yau limit, see Appendix \ref{app:exp}. For brevity, we will refer to such examples as examples where $X_0$ is Calabi-Yau, where the subscript zero refers to the $\a\rightarrow0$ limit. A less restrictive, but still sufficient condition for which this holds is \cite{delaOssa:2014cia}
\begin{equation}
\label{eq:nodil}
H^{(0,1)}(X)\cong H^{(0,2)}(X)=0\:,
\end{equation}
where the isomorphism is due to Serre-duality. We will assume that \eqref{eq:nodil} is true for the spaces we consider, but we note that as shown in Appendix \ref{app:exp}, this is true whenever $X_0$ is Calabi-Yau. 

Next consider the second term in Eq.\eqref{eq:delta2W}. As $\delta_1A$ is generic, it follows that we must require
\begin{equation*}
\delta_2(F\wedge\Omega)=0\:.
\end{equation*}
From this it follows that
\begin{equation}
\label{eq:exactFcond}
\Delta_2^a\wedge F_{a\bar b}\,\d z^{\bar b}=\bp_A\alpha_2\:,
\end{equation}
where $\alpha_2=(\delta_2A)^{(0,1)}$, and the operator $\bp_A$ is defined by
\[ \bp_A \beta= \bp\beta + A^{(0,1)}\wedge\beta - (-1)^p \beta\wedge A^{(0,1)}~,\]
for any $p$-form $\beta$, and where $A^{(0,1)}$ denotes the $(0,1)$ part of $A$.  Note that $\bp_A^2 = 0$. 
Eq.\eqref{eq:exactFcond} can be equivalently restated as the condition that
\begin{equation*}
\Delta_2^a\in\ker(\F)\:,
\end{equation*}
where the map
\begin{equation*}
\F\::\;\;\;H^{(p,q)}(TX)\rightarrow H^{(p,q+1)}({\rm End}(V))\: ,
\end{equation*}
is defined by contraction with $F$ as in Eq.\eqref{eq:exactFcond}
\[ {\cal F}(\Delta) =  F_{a\bar b}\,\d z^{\bar b}\wedge \, \Delta^a~.\]
We will show in a moment that by the Bianchi identity for $\cal F$, this map is indeed a map between cohomologies.

Similarly, the third term in the first line of Eq.\eqref{eq:delta2W}  gives the condition
\begin{equation}
\label{eq:exactRcond}
\Delta_2^a\wedge R_{a\bar b}\,\d z^{\bar b}=\bp_\nabla\kappa_2\:,
\end{equation}
where $\kappa_2$ is the $(0,1)$-part of the deformation of $\Theta$.  We also see that we need 
\begin{equation*}
\Delta^a\in\ker(\R)\:,
\end{equation*}
where the map
\begin{equation*}
\R\::\;\;\;H^{(p,q)}(TX)\rightarrow H^{(p,q+1)}({\rm End}(TX))\: ,
\end{equation*}
is given by contraction with $R$ as in Eq.\eqref{eq:exactRcond}. 

This discussion can also be given in terms of cohomologies, as first done by Atiyah \cite{MR0086359} and applied to the heterotic case in \cite{Anderson:2010mh, Anderson:2011ty, Anderson:2013qca, delaOssa:2014cia, Anderson:2014xha}. We give here a brief summary. Consider the operator
\begin{equation*}
\bp_E=\bp_{\tilde\nabla}+\R+\F\::\;\;\;\Omega^{(p,q)}(\g\oplus TX)\rightarrow\Omega^{(p,q+1)}(\g\oplus TX)\:,
\end{equation*}
where $\g=\End(TX)\oplus\End(V)$ and $\bp_{\tilde\nabla} = \bp_{\g} + \bp\,$ is a diagonal  holomorphic connection on the bundle $\g\oplus TX$.  To be clear about this notation,  the operator $\bp_E$, as a matrix acting on forms with values in $\g\oplus TX$, is defined as
\begin{equation*}
\bp_E = 
\left[
\begin{array}{ccc}
\bp_{\nabla}~ & ~0 &  ~\R\\
0 ~& ~\bp_{A} & ~\F\\
0~ & ~0 & ~\bp
\end{array}
\right]
=\left[
\begin{array}{cc}
\bp_{\g}~ & ~\R+\F\\
0~ & ~\bp
\end{array}
\right]
\end{equation*}
A short computation reveals that the operator $\bp_E$ satisfies 
\begin{equation*}
\bp_E^2=\bp_{\tilde\nabla}^2+  \bp_\nabla\R+ \R\bp + \bp_A\F +\F\bp =0\:,
\end{equation*}
 the due to the Bianchi identities for $F$ and $R$.  Indeed it is easy to check that 
 \[  \bp_A\F +\F\bp = \bp_A F~,\]
 with a similar equation for the Bianchi identity of $R$. Note that $\cal F$ and $\cal R$ are then maps between cohomologies.
The holomorphic connection $\bp_E$ therefore defines the bundle $E$ as a holomorphic extension bundle,
\begin{equation}
\label{eq:ses1}
0\rightarrow\g\rightarrow E\rightarrow TX\rightarrow0\:.
\end{equation}
Indeed, the holomorphic structure on the total space of a  holomorphic bundle over a complex manifold, can always be encoded in such a holomorphic extension bundle $E$. We are interested in computing infinitesimal simultaneous deformations of the complex structure on $X$ and the holomorphic structure on the bundles. Equivalently, we want to compute the infinitesimal deformations of $\bp_E$. These are computed by the cohomology
\begin{equation*}
T\M_{\bp_E}=H^{(0,1)}(E)\:,
\end{equation*}
which in turn can be computed by a long exact sequence in the cohomology of Eq.\eqref{eq:ses1}, yielding~\footnote{In the computation of Eq.\eqref{eq:H1Q1} we make the simplifying assumption that $H^0(TX)=0$. This assumption is fine when $X_0$ is Calabi-Yau, see Appendix \ref{app:exp}.}
\begin{equation}
\label{eq:H1Q1}
H^{(0,1)}(E)=H^{(0,1)}(\End(TX))\oplus H^{(0,1)}(\End(V))\oplus\ker(\R + \F)\:.
\end{equation}
Note then that in order for $\bp_E$ to remain holomorphic, thus preserving supersymmetry, we require $\Delta^a\in\ker(\R + \F)$, which is precisely the requirement obtained from the superpotential above. Indeed, the infinitesimal moduli are now given by
\begin{equation*}
x_2=(\Delta_2^a,\alpha_2,\kappa_2)\in H^{(0,1)}(E)\:,
\end{equation*}
satisfying the condition $\bp_E\: x_2=0$. 

\subsection{Anomaly Deformations}
We next consider the second line of Eq.\eqref{eq:delta2W}. Writing the first term out, we get
\begin{equation*}
\int_X\delta_2(H+i\d\omega)\wedge\delta_1\Omega=\int_X\left(\frac{\a}{2}(\tr\:\alpha_2\wedge F-\tr\:\kappa_2\wedge R)+\d\tau_2\right)\wedge\delta_1\Omega\:.
\end{equation*}
The second term of the second line of Eq.\eqref{eq:delta2W} is given by
\begin{equation*}
\int_X(H+i\d\omega)\wedge\delta_2\delta_1\Omega=2i\int_X\p\omega\wedge\delta_2\delta_1\Omega\:.
\end{equation*}
Noting that $H + i\, \omega = 2\, i\, \partial \omega$, 
 it is clear that only the $(1,2)$-part of $\delta_2\delta_1\Omega$ contributes 
\begin{equation*}
(\delta_2\delta_1\Omega)^{(1,2)}= \Delta^a_1\wedge\chi_{2ab\bar c}\,\d z^{b\bar c}
= \Delta_1^a\wedge\Delta_2^b\, \Omega_{a b c}\, \dd x^c\:.
\end{equation*}
Using this, we can rewrite
\begin{align*}
2i\int_X\p\omega\wedge\delta_2\delta_1\Omega
&= 2 i\int_X \p_{[a}\omega_{b]\bar c}\:\dd z^{\bar c}\: \wedge \Delta_1^d
\:\Omega_{def}\:\wedge \Delta_2^{[e}\wedge \dd z^{f]ab}
\\
&= - 2 i\int_X \p_{[a}\omega_{b]\bar c}\:\dd z^{\bar c}\: \wedge \Delta_1^d
\:\Omega_{def}\:\wedge \Delta_2^{[a}\wedge \dd z^{b]ef}
\\
&=-4i\int_X\Delta_2^a\wedge\p_{[a}\omega_{b]\bar c}\:\d z^{b\bar c}\wedge\chi_1\:,
\end{align*}
where in the second line we have used
\[ 0 = 2\, \Delta^{[e}\wedge\dd z^{abf]} = \Delta^{[e}\wedge\dd z^{f]ab} + \Delta^{[a}\wedge\dd z^{b]ef}~.\]
Putting it all together, and requiring $\delta_1\Omega$ generic, we find that we need
\begin{align*}
\bp\tau_2^{(0,2)}&=0\\
-4\Delta^a_2\wedge i\p_{[a}\omega_{b]\bar c}\:\d z^{b\bar c}+\frac{\a}{2}\big(\tr\:(\alpha_2\wedge F)-\tr\:(\kappa_2\wedge R)\big)+\p\tau_2^{(0,2)}+\bp\tau_2^{(1,1)}&=0\:.
\end{align*}
The first equation together with \eqref{eq:nodil} imply that $\tau^{(0,2)}$ is $\bp$-exact, that is
\begin{equation*}
\tau_2^{(0,2)}=\bp\beta^{(0,1)} ,
\end{equation*}
for some $(0,1)$-form $\beta^{(0,1)}$.
The second equation then gives the following condition
\begin{equation}
\label{eq:conddelta2Wanomaly}
-4\Delta^a_2\wedge i\p_{[a}\omega_{b]\bar c}\: \d z^{b\bar c}+\frac{\a}{2}\big(\tr\:(\alpha_2\wedge F)-\tr\:(\kappa_2\wedge R)\big)+\bp\tau_2^{(1,1)}-\bp\p\beta^{(0,1)}=0\:,
\end{equation}
which can be rewritten as 
\begin{equation}
- \H(x_2)_a\d z^a= \frac{1}{2}\bp\left(\tau_2^{(1,1)}-\p\beta^{(0,1)}\right)\:,
\label{eq:Hconst}
\end{equation}
where $\cal H$ is the map defined by
\begin{equation}
\label{eq:mapH}
\H=\hat H+ \F+ \R\::\;\;\;\Omega^{(p,q)}(E)\rightarrow \Omega^{(p,q+1)}(T^*X) ,
\end{equation}
and 
\begin{equation*}
\begin{split}
\H(x)_b&=  \Delta^a\wedge\hat H_{ab\bar c}\:\d z^{\bar c} 
 - \frac{\a}{4}\Big(\tr\:(\alpha\wedge F_{b\bar c}\d z^{\bar c})
-\,\tr\:(\kappa\wedge R_{b\bar c}\d z^{\bar c})\Big)\:,\\[5pt]
\hat H_{ab\bar c}\:\d z^{\bar c} &= 2 i\p_{[a}\omega_{b]\bar c}\:\d z^{\bar c} 
= H ^{(2,1)}_{ab\bar c}\:\d z^{\bar c} \:.
\end{split}
\end{equation*}

We have extended the definition of the maps $\F$ and $\R$ to forms $x = (\kappa,\alpha,\Delta)$ with values in $E$.
In fact,  $\F$ and $\R$ are understood as acting on both $TX$-valued forms as before, and on $\g$-valued forms by the trace on the endomorphism bundles.  That is, we are extending the definition of these maps so that 
\begin{align*}
{\cal F}_b (\alpha) &= ~~ \frac{\a}{4}\: \tr ( F_{b\bar c}\:\dd z^{\bar c}\wedge \alpha)~,\quad
\alpha\in \Omega^{(p,q)}({\rm End}(V))~,\\[3pt]
{\cal R}_b (\kappa) &= - \frac{\a}{4}\: \tr ( R_{b\bar c}\:\dd z^{\bar c}\wedge \kappa)~,\quad
\kappa\in \Omega^{(p,q)}({\rm End}(TX))~.
\end{align*}
Alternatively the pre-factors $\pm\frac{\a}{4}$ could be pulled into a re-definition of the trace on $\g$. Note the different sign of the action of $\R$ relative to the action of $\F$.  Altogether, these maps act as follows \begin{align*}
 \F(x) = 
 \left[
 \begin{array}{cc}
 &\F(\kappa) \\ & \F_a(\alpha) \\ & \F(\Delta)
 \end{array}
 \right]
 =
 \left[
 \begin{array}{cc}
 & 0\\
 & \frac{\a}{4}\, \tr(F_a\wedge\alpha)\\
 & \F_{a\bar b}\:\dd x^{\bar b} \wedge\Delta^a
 \end{array}
 \right]
 ~,\qquad
 \R(x) = 
 \left[
 \begin{array}{cc}
 &\R_a(\kappa) \\ & \R(\alpha) \\ & \R(\Delta)
 \end{array}
 \right]
 =
 \left[
 \begin{array}{cc}
 &- \frac{\a}{4}\, \tr(R_a\wedge\alpha)\\
 & 0\\
 & R_{a\bar b}\:\dd x^{\bar b} \wedge\Delta^a
 \end{array}
 \right]
 ~,
 \end{align*}
where $F_a=F_{a\bar b}\d z^{\bar b}$ and $R_a=R_{a\bar b}\d z^{\bar b}$.
We will see below that the map $\cal H$ is in fact a map between cohomologies.   Hence, we see that the  equation for moduli \eqref{eq:Hconst} for  $x_2$ can be equivalently stated as 
\[ x_2 = \ker({\cal H})~.\] 
This of course is in agreement with what was found from the ten-dimensional supergravity perspective in \cite{delaOssa:2014cia, Anderson:2014xha}.

\subsection{The Map $\H$, and Extending $E$ by $T^*X$}
As will be discussed in Section \ref{sec:Dterm}, there is also an additional constraint coming from the requirement that the bundles involved satisfy the Yang-Mills condition
\begin{align*}
\omega^{mn}F_{mn}&=0\\
\omega^{mn}R_{mn}&=0\:.
\end{align*}
These constraints are manifest in the 4d theory as D-terms. These constraints also appear naturally when we consider the infinitesimal moduli space in terms of cohomologies, as shown in \cite{delaOssa:2014cia}, and part of Section \ref{sec:defD} together with Section \ref{sec:Dterm} is a review of this nice result. In Section \ref{sec:defD} we also consider in more detail the map $\H$, and derive a theorem relating the number of complex structure and bundle moduli which are lifted.

The map $\cal H$ defined in Eq.\eqref{eq:mapH} can be shown to give a well defined map in cohomology if and only if the Bianchi identity 
\begin{equation}
\label{eq:BI}
\d H=\frac{\a}{4}(\tr\:F^2-\tr\:R^2) ,
\end{equation}
is satisfied \cite{delaOssa:2014cia, Anderson:2014xha}. This fact is equivalent to the map $\H$ commuting with the respective cohomology operators
\begin{equation}
\label{eq:clousureH}
\bp\H+\H\,\bp_E = 0\:.
\end{equation}
In fact one obtains
\begin{align*}
\bp(\H_b(x))+\H_b\,(\bp_E\, x)&=
- \frac{1}{4} \Big(\dd H - \frac{\a}{4} (\tr F^2 - \tr R^2)\Big)_{ab\bar c\bar d}~\,
\dd z^{\bar c}\wedge\dd z^{\bar d}\wedge\Delta^a
\\[5pt]
&\quad + \frac{\a}{4} \Big(\tr \big(\bp_A (F_{b\bar c}\dd z^{\bar c})\wedge \alpha\big)
- \tr\big(\bp_\nabla (R_{b\bar c}\dd z^{\bar c})\wedge \kappa\big) \Big)\:,
\end{align*}
where the first line vanishes due to the heterotic Bianchi identity \eqref{eq:BI}, while the second line vanishes due to the Bianchi identities for the curvatures. As also observed in \cite{Anderson:2014xha}, we note that if we view the extension map as bundle-valued form $\H\in\Omega^{(0,1)}(E^*\otimes T^*X)$, then \eqref{eq:clousureH} implies that $\bp_E\H=0$. In other words, $\H$ represents an element of the cohomology
\begin{equation*}
{\rm Ext}^1(E,T^*X)=H^{(0,1)}(E^*\otimes T^*X)\:,
\end{equation*}
which contains the allowed extension classes of $E$ by $T^*X$. This particular extension group was computed in \cite{Anderson:2014xha}. Note that extension classes differing by an exact element of ${\rm Ext}^1(E,T^*X)$ give rise to holomorphically equivalent extensions.

We remark that $\H$ represents a {\it particular} extension class in ${\rm Ext}^1(E,T^*X)$. Indeed, let us consider the operator
\begin{equation*}
\bar D= \bp_{\tilde\nabla}+\C+\hat H\::\;\;\;
\Omega^{(p,q)}(T^*X\oplus E )\rightarrow\Omega^{(p,q+1)}(T^*X\oplus E)\:,
\end{equation*}
where we, for brevity, we have defined
\begin{equation*}
\C=\F+\R ,
\end{equation*}
and $\bp_{\tilde\nabla}$ has been extended to a holomorphic diagonal  connection on the bundle $T^*X\oplus\g\oplus TX$. Written out in matrix form, $\bar D$ reads
\begin{equation}
\bar D\;=\; \left[ \begin{array}{cc}
\bp & \,\H \\
0  & \,\bp_E  \end{array} \right]=
\left[ \begin{array}{ccc}
\bp & \,\C_\g & \,\hat H \\
0 & \,\bp_{\g} & \,\C_{TX} \\
0 & 0 & \bp \end{array} \right]\:,
\label{eq:barD}
\end{equation}
which acts on forms with values in $(T^*X,\g,TX)$, and where, as before, $\bp_{\g}$ is the holomorphic connection on the bundle $\g$. Unless it is clear from the context, we will also denote the map $\C$ as $C_{TX}$ or $\C_\g$ depending on whether it acts on tangent bundle or bundle valued forms respectively. Note that $\C_{TX}$ corresponds to the original Atiyah map.

We now compute  $\bar D^2$ which gives
\begin{equation*}
\bar D^2= 
\left[ \begin{array}{ccc}
~\bp^2 & ~~\bp \C_\g + \C_\g \bp_\g & ~~~\bp\hat H + \hat H\bp + \C_\g\wedge\C_{TX} \\
~0 & ~~\bp_\g^2 & ~~~ \bp_\g\C_{TX} + \C_{TX}\bp \\
~0 & ~~0 & ~~~\bp^2 \end{array} \right]\:,
\end{equation*}
which again vanishes using $\hat H= H^{(2,1)}$, the Bianchi identity on the curvature of $\g$, and the Bianchi Identity \eqref{eq:BI}. 
In particular,
\[ \bp\hat H + \hat H\bp + \C_\g\wedge\C_{TX} = 
\dd H - \frac{\a}{4}\, (\tr F^2 - \tr R^2) = 0~.\]
The operator $\bar D$ defines the holomorphic bundle $\Q$ as an extension
\begin{equation*}
0\rightarrow T^*X\rightarrow\Q\rightarrow E\rightarrow0\:,
\end{equation*}
with extension class $\H$.

We have just seen that the Bianchi identity is responsible for the choice of  a particular extension class $\H$. Furthermore, looking at Eq.\eqref{eq:barD}, we see that we can equivalently define $\Q$ by first extending $\g$ by $T^*X$ with extension class $\C_\g$,
\begin{equation*}
0\rightarrow T^*X\rightarrow \tilde E \rightarrow\g\rightarrow0\:.
\end{equation*}
Note that this is precisely the bundle $\tilde E=E^*$.
We then extend $TX$ by this bundle to get
\begin{equation*}
0\rightarrow E^*\rightarrow\Q\rightarrow TX\rightarrow0\:,
\end{equation*}
with extension class again given by $\H$. Note that this is precisely the construction of $\Q^*$. It follows that the extension class $\H$ is chosen so to make $\Q$ {\it self-dual} as a holomorphic bundle. This can also be seen from the definition \eqref{eq:barD} of the holomorphic operator $\bar D$. It is simply the observation that the maps $\C_{TX}$ and $\C_\g$ are given by the same class as elements of the extension group ${\rm Ext}^1(T^*X, \g)={\rm Ext}^1(\g,TX)=H^{(0,1)}(T^*X\otimes\g)$.

\subsection{On Deformations of $\bar D$ and $\ker(\H)$}
\label{sec:defD}
We now consider infinitesimal deformations of $\bar D$, and compare to what was found by deforming the superpotential. This calculation was done in \cite{delaOssa:2014cia, Anderson:2014xha}, and we quote the results here. It is important to remark that for infinitesimal deformations we do not need to worry about spoiling the self-duality condition, as this is a condition on the extension classes $\{\C,\H\}$, whose cohomology classes do not change under infinitesimal deformations. The infinitesimal moduli space of $\bar D$ is given by
\begin{equation}
\label{eq:ModuliD}
T\M_{\bar D}=H^{(0,1)}(\Q)=\left(H^{(0,1)}(T^*X)\big/{\rm Im}(\F)\right)\oplus\ker(\H)\:,
\end{equation}
where $\ker(\H)\subseteq H^{(0,1)}(\Q_1)$. As expected, this includes the kernel structure which we previously derived from superpotential deformations, that is $x_2\in\ker(\H)$. The $H^{(0,1)}(T^*X)$-part may be interpreted as hermitian moduli, that is $\bp$-closed forms we can add to $\tau_2^{(1,1)}$ without changing Eq.\eqref{eq:conddelta2Wanomaly}. We refer to these as $\tau_0^{(1,1)}$ so that
\begin{equation}
\label{eq:compKahler}
\bp\tau_0^{(1,1)}=0\:,\;\;\;\Rightarrow\;\;\;\tau_0^{(1,1)}\in H^{(1,1)}(X)\cong H^{(0,1)}(TX)\:.
\end{equation}
The reader may wonder if it is appropriate to mod out by $\bp$-exact forms in order to arrive at hermitian moduli valued in the Dolbeault cohomology $H^{(1,1)}(X)$. Moreover, in the final result, the hermitian moduli are also quotiented by ${\rm Im}(\F)$. In Section \ref{sec:Dterm} we will see why this cohomology is correct upon considering the appropriate symmetries of the system. We will also see how the quotient by ${\rm Im}(\F)$ appears when we look at preservation of the Yang-Mills condition and the Green-Schwarz Mechanism \cite{green1984anomaly}. 

Before we move to discuss symmetries and D-terms, it is instructive to investigate the kernel structure of $\H$ a bit further. 
Let
\[ x_t = \left[\begin{array}{c}\gamma_t\\ \Delta_t \end{array}\right] \in \Omega^{(0,1)}(\g\oplus TX)
~, \qquad {\rm and}\quad y_t\in \Omega^{(0,1)}(T^*X)~,\]
where $\gamma_t\in \Omega^{(0,1)}(\g)$ and $\Delta_t\in\Omega^{(0,1)}(TX)$. 
 Then the equations for the moduli of the Strominger system are
\begin{align*}
\bp y_t &= - {\cal H}(x_t)~,\\
\bp_E x_t & = 0~.
\end{align*}
The variable $t$ denotes a generic infinitesimal deformation. We let $t=\{a,b,..\}$, $t=\{i,j,..\}$ or $t=\{x,y,..\}$ depending on if it corresponds to a complex structure, bundle, or hermitian modulus respectively. It should be clear from the context if an index $a$ corresponds to a complex structure modulus or a holomorphic index.

Suppose that $t$ is a deformation which is not a deformation of the hermitian structure, that is, it is a parameter of the holomorphic structure on $E$. In particular, these equations split as
\begin{align}
\bp y_a &= - {\cal H}(x_a) = - \hat H (\Delta_a) - {\cal C}_\g(\gamma_a) = - {\cal H}_\Delta(x_a)~,\label{eq:HDelta}\\
\bp y_i &= - {\cal H}(x_i) = - {\cal C}_\g(\gamma_i)~,\label{eq:Hg}
\end{align}
where $\Delta_a\in \ker\C$ and $\gamma_i\in H^{(0,1)}(\g)$. The first equation depends only on the complex structure moduli $\Delta_a$ (hence the label $\H_\Delta$).  This is easy to see by recalling that the $\gamma_a$ are uniquely determined by the $\Delta_a$ from 
\[ \bp_{\g}\gamma_a = - \C(\Delta_a)~.\]
In fact, there is an ambiguity in obtaining $\gamma$ from this equation, which is that we can add to it a $\bp_{\g}$-closed part.  However these correspond to bundle moduli already accounted for in those elements in $H^{(0,1)}(\g)$. 

In any case, the first equation represents a further lifting of the complex structure moduli so that the extra lifted complex structure parameters live in
\[  L_\Delta = \ker({\cal C}_{TX}) \setminus\ker({\cal H}_\Delta)~,\]
where $\ker({\cal C}_{TX}) \setminus\ker({\cal H}_\Delta)$ denotes the subtraction of $\ker({\cal C}_{TX})$ by $\ker({\cal H}_\Delta)$. The second is a lifting of the bundle moduli so that the lifted bundle moduli live in
\[ L_{\gamma} = H^{(0,1)}(\g)\setminus \ker({\cal C}_\g)~.\]
The $y_a$ and $y_i$ are the necessary deformations of the complexified hermitian structure due to a change of the holomorphic structure on $\cal Q$ so that the  deformed bundle $\Q$ is holomorphic, and the ambiguity of a $\bp$-closed form in each case corresponds to moduli already counted in $H^{(0,1)}(T^*X)$.  

We now prove an interesting theorem for the number of lifted moduli.

\begin{Theorem}
\label{tm:ImH2=ImF}
The number of complex structure moduli lifted by $\C_{TX}$ equals the number of bundle moduli lifted by $\C_{\g}$. In other words,
\begin{equation*}
{\rm dim}\big({\rm Im}(\C_{TX})\big)=\rm{dim}\big({\rm Im}(\C_{\g})\big)\:.
\end{equation*}
\end{Theorem}

\begin{proof}
Suppose that we consider $\gamma_i\in H^{(0,1)}(\g)$ and 
${\cal C}_\g(\gamma_i)\in {\rm Im}({\cal C}_\g)$.  We can think of $\C_\g(\gamma_i)$ as a $(1,2)$-form on $X$ (instead of a $(0,2)$ with values in the holomorphic cotangent bundle).  Let $\chi\in H^{(2,1)}(X)$ and consider the integral
\[ \int_X {\cal C}_\g(\gamma_i)\wedge \chi~.\]
We claim 
\[(\C_\g(\gamma_i),*\bar\chi) = \int_X {\cal C}_\g(\gamma_i)\wedge \chi = 0~, \quad \forall~\chi\in H^{(2,1)}(X)
\qquad\iff\qquad \gamma_i\in\ker({\cal C}_\g)~.\]
Indeed, if 
${\cal C}_\g(\gamma_i)=\bp y_i$ then, by integration by parts, the integral vanishes for every 
$\chi\in H^{(2,1)}(X)$.  On the other hand if the integral vanishes for every $\chi\in H^{(2,1)}(X)$ then, as
$\bp {\cal C}_\g(\gamma_i) = - {\cal C}_\g(\bp_\g(\gamma_i)) = 0$, and as $\{*\bar\chi\}$ spans $\mathcal{H}^{(1,2)}(X)$ for $\chi\in\mathcal{H}^{(2,1)}(X)$, it must be the case that ${\cal C}_\g(\gamma_i)$ is $\bp$-exact. Here $\mathcal{H}^{(*,*)}(X)$ denotes harmonic forms. Consequently, for $\gamma_i\notin \ker({\cal C_\g})$, there exists a $\chi\in H^{(2,1)}(X)$ such that 
\[ \int_X {\cal C}_\g(\gamma_i)\wedge \chi \ne 0~.\]
But for this integral to vanish it must be the case that $F\wedge\chi$ or equivalently $\F(\Delta)$ is non-trivial in cohomology. Here $\Delta\in H^{(0,1)}(TX)$ is the complex structure modulus correspoding to $\chi$. This proves that for every bundle modulus lifted by ${\cal C}_\g$, there is at least one complex structure modulus lifted by
${\cal C}_{TX}$.  Another way to state this result is to say that
\[ \dim({\rm Im}({\cal C}_{TX})) \ge \dim({\rm Im}({\cal C}_{\g}))~,\]
where we have used the  isomorphisms
\[ H^{(0,1)}(TX) \setminus \ker({\cal C}_{TX}) \simeq {\rm Im}({\cal C}_{TX})~,
\qquad
H^{(0,1)}(\g) \setminus \ker({\cal C}_{\g}) \simeq {\rm Im}({\cal C}_{\g})~.\]

Similarly, by considering the claim
\[ \int_X {\cal C}_\g(\gamma)\wedge \chi_a = 0~, \quad \forall~\gamma\in H^{(0,1)}(\g)
\qquad\iff\qquad \Delta_a\in\ker({\cal C}_{TX})~,\]
for $\chi_a\in  H^{(2,1)}(X)$ corresponding to $\Delta_a$, 
one can prove that 
\[ \dim({\rm Im}({\cal C}_{TX})) \le \dim({\rm Im}({\cal C}_{\g}))~.\]

We conclude  that
\[ \dim({\rm Im}({\cal C}_{TX})) = \dim({\rm Im}({\cal C}_{\g}))~,\]
or equivalently that
\[ L_\gamma\simeq H^{(0,1)}(TX) \setminus \ker({\cal C}_{TX})~.\]

\end{proof}
As an aside we note that for a given lifted complex structure modulus $\Delta\in H^{(0,1)}(TX)$, the corresponding bundle modulus $\gamma\in H^{(0,1)}(\g)$ is given by the Serre dual of 
\begin{equation*}
\C(\Delta)\in H^{(0,2)}(\g)\:,
\end{equation*}
where Serre-duality is given by the holomorphic three-form $\Omega$. It should also be noted that the self-duality of $\Q$ is important for the theorem to hold, which is easy to see as this implies that $\C_{TX}$ and $\C_\g$ are given by the same curvature terms.

We should bear in mind that as the right hand side of  equation \eqref{eq:HDelta}, which depends only on the variations $\Delta_a$ of the complex structure, this equation represents a further lifting of the complex structure moduli. We also note that even in the Calabi-Yau case, where $\hat H=0$ (when $\a=0$), this map need not be non-trivial, and so may lead to lifting extra complex structure moduli not lifted by the Atiyah map $\C_{TX}$. We will return to this in more detail later in the paper. I particular when discussing obstructions in Section \ref{sec:yuk}.


\section{Symmetries, D-terms and the Green-Schwarz Mechanism}
\label{sec:Dterm}

Up until now we have been assuming that if a modulus is closed with respect to some holomorphic operator, then the true modulus corresponds to an element of the cohomology given by the holomorphic operator. This is true for the complex structure and bundle moduli, but we are not yet justified in claiming that the complexified K\"ahler moduli given by \eqref{eq:compKahler} take values in $H^{(0,1)}(T^*X)$. We begin this section by showing that this is indeed the case.

\subsection{Symmetries of the Hermitian Moduli}
Let us first consider the imaginary part of the hermitian moduli $\tau_0^{(1,1)}$ given by the $(1,1)$-part of the deformation of the hermitian form $(\delta_0\omega)^{(1,1)}$. Assume that 
\begin{equation*}
\bp\big((\delta_0\omega)^{(1,1)}\big)=0\:.
\end{equation*}
As we shall see below in Section \ref{sec:YM}, the manifold $X$ is also required to be conformally balanced,
\begin{equation}
\label{eq:confbal}
\d\left(e^{-2\phi}\omega\wedge\omega\right)=\d\left(\hat\omega\wedge\hat\omega\right)=0\:,
\end{equation}
where $\hat\omega=e^{-\phi}\omega$ is referred to as the Gauduchon hermitian form. A straightforward but somewhat tedious computation now shows that the preservation of the conformally balanced condition \eqref{eq:confbal} requires that \cite{delaOssa:2014cia, Svanes:2014ala}
\begin{equation}
\label{eq:defconfbal}
\bp^{\tilde\dagger}\left((\delta_0\omega)^{(1,1)}\right)=\bp^{\tilde\dagger}\left(\omega\:\delta_0\phi+\p^{\tilde\dagger}\Lambda^{(2,1)}_0\right)\:,
\end{equation}
where the adjoints are taken with respect to the metric $\tilde g=e^{-2\phi}g$.  This also uses the fact that for a generic deformation of the hermitian form
\begin{equation*}
\delta\hat\omega_0=\lambda_0\:\hat\omega+\hat h_0\:,
\end{equation*}
where the two-form $\hat h_0$ is primitive and $\lambda_0$ is a function.  On a conformally balanced manifold $X$,  is an easy exercise left to the reader to show that, using diffeomorphisms, $\lambda_0$ can in fact be set to a constant. Note that since we assume that
\begin{equation*}
H^{(2,0)}(X)\cong H^{0}(TX)=0\:,
\end{equation*}
where the isomorphism is given by $\Omega$, it follows by the Hodge decomposition of $\hat h^{(2,0)}_0$ that 
\begin{equation*}
\hat h^{(2,0)}_0=\bp^{\hat\dagger}\Lambda^{(2,1)}_0\:.
\end{equation*}
This equation defines the $(2,1)$-form $\Lambda_0$. As $\hat h^{(2,0)}_0$ is determined by the deformation of the complex structure, then so is $\Lambda^{(2,1)}_0$. It therefore follows from \eqref{eq:defconfbal} that the $\bp$-exact part of $(\delta_0\omega)^{(1,1)}$ is determined by the deformations of the complex structure and dilaton. We are therefore correct to mod out by $\bp$-exact forms for the imaginary part of the complexified hermitian moduli \eqref{eq:compKahler} corresponding to deformations of the hermitian form $(\delta_0\omega)^{(1,1)}$.

Next we turn to the real part of the complexified hermitian modulus $\tau_0^{(1,1)}$. We first compute the real part of the two-form $\tau_2$ corresponding to a generic symmetry transformation
\begin{equation}
\label{eq:genericGauge}
\left({\rm Re}(\tau_2)\right)_{\delta}=\d\lambda-\frac{\a}{2}\left(\tr\,F\epsilon_1-\tr\,R\epsilon_2\right)+\frac{\a}{4}\d\left(\tr\,A\epsilon_1-\tr\,\Theta\epsilon_2\right)\:,
\end{equation}
where $\d\lambda$ is an exact form corresponding to the allowed infinitesimal re-parametrisations of the $B$-field, while $\epsilon_1$ and $\epsilon_2$ correspond to infinitesimal gauge-transformations of $A$ and $\Theta$ respectively. Note that in computing \eqref{eq:genericGauge} we also need to consider deformations coming from the gauge-transformation of the $B$-field, required by the Green-Schwarz mechanism \cite{green1984anomaly}
\begin{equation*}
\delta B=-\frac{\a}{4}\Big(\tr\:(\d A\epsilon_1)-\tr\:(\d\Theta\epsilon_2)\Big)\:.
\end{equation*}
At this point, we are not concerned with transformations changing the holomorphic structures on the bundles, so we assume that $\bp_A\epsilon_1=\bp_\nabla\epsilon_2=0$. The tangent bundle $\End(TX)$ is assumed to be irreducible. As we will see in the next section and in Section \ref{sec:EndTconn}, $\End(TX)$ is also required to be a stable bundle, while $\End(V)$ is poly-stable. It follows that $\bp_\nabla\epsilon_2=0$ implies that $\epsilon_2=0$. Moreover, for a poly-stable bundle $V=\oplus_i V_i$ we may assume that
\begin{equation*}
\epsilon_1=\sum_ic_i I_i\:,
\end{equation*}
where the $I_i$'s are identity matrices and the $c_i\in\mathbb{C}$ are chosen so that $\tr\:\epsilon_1=0$. For such gauge transformations, \eqref{eq:genericGauge} reduces to
\begin{equation}
\label{eq:genericGauge2}
\left({\rm Re}(\tau_2)\right)_{\delta}=\d\lambda-\frac{\a}{4}\sum_ic_i\,\tr\,F_i\:,
\end{equation}
where $F_i$ are the curvatures of the individual bundles $V_i$. Consider the first term of the right hand side of \eqref{eq:genericGauge2}. The last term will be discussed below. The $(1,1)$-part of the first term reads 
\begin{equation*}
\d\lambda^{(1,1)}=\bp\lambda^{(1,0)}+\p\lambda^{(0,1)}\:.
\end{equation*}
We hence see that the symmetries we should mod out by when considering the real part of the hermitian moduli $\left({\rm Re}(\tau_0)\right)^{(1,1)}$ also includes $\bp$-exact forms. One might worry about the $\p$-exact terms in the above expression. That is, could there be $\bp$-closed $\p$-exact terms in our hermitian moduli that should not be considered as they are part of the symmetries we mod out by? The answer to this is no. To see why, let $u\in\Omega^{(0,1)}(X)$ and assume that $u$ satisfies
\begin{equation*}
\bp\p u=0\:.
\end{equation*}
Then $\bp u$ will define an element of
\begin{equation*}
H_\p^{(0,2)}(X)\cong H_{\bp}^{(2,0)}(X)=0\:,
\end{equation*}
It follows that $\bp v=u$ for some function $v$. But since we are also assuming that $H^{(0,1)}(X)=0$, it follows that $u=\bp f$ for some function $f$, and hence $\p u$ is $\bp$-exact and therefore not part of the hermitian moduli. Hence we are correct in viewing the complexified hermitian moduli as elements of $H^{(0,1)}(T^*X)$.

Note also that we should mod out the hermitian moduli in \eqref{eq:ModuliD} by the image of $\F$, which we did not see from varying the superpotential. It should be noted that this quotient is a natural part of the gauge-structure of $\Q$, whereby modding out by $\bar D$-exact terms means modding out by terms of the form
\begin{equation}
\label{eq:gaugeD}
\bar Dz=(\bp b+\C_\g(\epsilon)+\hat H(\delta),\:\bp_{\g}\epsilon+\C_{TX}(\delta),\:\bp\delta)\:,
\end{equation}
where $z=(b,\epsilon,\delta)\in\Omega^0(\Q)$, that is $b\in H^0(T^*X)$, $\epsilon_1+\epsilon_2=\epsilon\in H^0(\g)$ and $\delta\in H^0(TX)$. Note then that $\bp\delta$ corresponds to trivial deformations of the complex structure, and we may set $\bp\delta=0$ without loss of generality.  Since $H^0(TX)\cong H^{2,0}(X)=0$ this implies that $\delta=0$. Furthermore, $\bp_{\g}\epsilon$ corresponds to trivial deformations of the holomorphic vector bundles. Setting $\bp_{\nabla}\epsilon=0$ then implies $\epsilon_\nabla=0$ as we saw above. However, $\bp_{\g}\epsilon_A = 0$ does not imply $\epsilon_A=0$, as $V$ is only poly-stable in general. In this case, Eq.\eqref{eq:gaugeD} requires that we should mod out by the image of $\F$, as the elements in ${\rm Im}{(\cal F)}$ precisely correspond to gauge transformations of this kind. This quotient is also understood for the real part of the complexified hermitian moduli by the gauge transformation of $\tau_2$ and the $B$-field, giving rise to the second term on the right hand side of \eqref{eq:genericGauge2}. For the imaginary part of the hermitian moduli, it can be understood as preserving the preserving the Yang-Mills condition of $A$  \cite{delaOssa:2014cia}. Moreover, it is related to D-term conditions in the four-dimensional theory which we discuss next.

\subsection{Yang-Mills Equations and D-terms}
\label{sec:YM}
In this section we discuss the Yang-Mills equations in the ten-dimensional supergravity theory and their relation to D-terms in the effective four-dimensional theory.

In the ten-dimensional theory, in order for the supersymmetric variation of the gaugino fields to vanish, the gauge bundle should be holomorphic with a Yang-Mills connection
\beq
\label{eq:YM}
g^{a \bar{b}} F_{a \bar{b}} =0 .
\eeq
The ${\rm{SU}}(3)$-structure of a heterotic compactification guarantees that the compactification admits a balanced Gauduchon metric, $\hat g=e^{-\phi}g$ where
\begin{equation}
\label{eq:confbal2}
\d(\hat\omega\wedge\hat\omega)=0\:.
\end{equation}
With respect to this metric, the slope for a vector bundle $E$ on $X$ is defined as
\beq
\mu\left(E\right) = \frac{1}{{\rm{rk}}\left(E\right)} \int_X c_1\left(E\right) \wedge \hat{\omega}^2\:.
\eeq
Recall that the gauge-connection $A$ is holomorphic. We will further require the gauge connection $A$ to be hermitian. It was shown in  \cite{donaldson1985anti, uhlenbeck1986existence, Yau87, MR939923, MR915839} that the existence of a hermitian Yang-Mills connection is equivalent to demanding that the bundle is \emph{poly-stable}. A bundle $V_i$ is stable if 
\beq
\mu\left(E \right) < \mu\left(V_i\right) ,
\eeq
for all sub-sheaves $E$ of $V_i$. Further, a bundle that can be written as a direct sum
\begin{equation}
\label{eq:split}
V=\oplus_iV_i\:,
\end{equation}
is poly-stable if each $V_i$ is stable with the same slope $\mu\left(V_i\right) = \mu\left(V\right)$. The slope vanishes for the total bundle $V$ in the case of heterotic string compactifications because the structure group of $V$ is contained in $E_8\times E_8$.

We now discuss how the conformally balanced constraint \eqref{eq:confbal2} on the hermitian form appears from the $D$-terms of the effective theory.   In the four-dimensional theory, we conjecture that there is a D-term potential of the form 
\beq
\label{eq:potD}
V_D\sim\sum_i\mu(V_{i})^2.
\eeq
Demanding that this vanishes of course leads back to the zero-slope condition. Such a potential is exactly of the same form as it appears in Calabi-Yau compactifications, where it can be derived from the term 
\begin{equation*}
\a\int_{\mathcal{M}_{10}}\sqrt{-g}\:\tr\left[(g^{MN}F_{MN})^2\right]\;\;\in\;S_{10} \;,
\end{equation*}
in the ten-dimensional action \cite{Anderson:2009nt, Anderson:2009sw, Anderson:2011cza} (the indices $M, N$ in this equation are ten-dimensional indices). It is reasonable to expect a similar potential in the torsional case,  modified only by the definition of the slope. We note first that this term produces a potential which vanishes iff
\begin{equation}
\label{eq:YMD}
\omega\lrcorner F = 0~,
\end{equation}
when we take the four-dimensional space-time is Minkowski.  
Further, we will shortly see that  the potential \eqref{eq:potD} precisely forces the hermitian moduli to be orthogonal to $F_i\,$, as derived above and in \cite{delaOssa:2014cia}.

The moduli are obtained from the fluctuations around Eq.\eqref{eq:potD}. Considering the mass-matrix for $V_D$ for the bundle $V$  and finding the second order deformation as before we have
\begin{equation*}
\delta_1\delta_2V_D\sim \sum_i\delta_1\left(\mu(V_i)\right)\delta_2\left(\mu(V_i)\right)\:.
\end{equation*}
Therefore the massless modes satisfy
\begin{equation*}
\delta(\mu(V_i))=0\;,\;\; \forall\:i\:,
\end{equation*}
leading to
\begin{equation*}
\int_X\d(\tr\:\delta A_i)\wedge\hat\omega\wedge\hat\omega+2\:\int_X\tr\:F_i\wedge\hat\omega\wedge\delta\hat\omega=0\:, \;\; \forall\:i\:.
\end{equation*}
The first of these terms will generate a potential unless 
\[\d(\hat\omega\wedge\hat\omega)=0~,\]
 which is the conformally balanced condition. The second term generates a potential unless
\begin{equation*}
(\tr\:F_i,\delta\hat\omega)=0\:.
\end{equation*}
In \cite{delaOssa:2014cia} it was shown from a ten-dimensional perspective that this is exactly the same condition required on the imaginary part of the compexified hermitian moduli, $\delta\hat\omega$, in order to preserve the Yang-Mills condition \eqref{eq:YMD}. This then also explains also why the imaginary part should be modded out by $\textrm{Im}(\F)$, completing our understanding of why this quotient appears in \eqref{eq:ModuliD}.

It should be noted that if we allow deformations of the bundle away from the locus where the bundle is a direct sum of the form Eq.\eqref{eq:split}, referred to as a stability wall in the hermitian moduli space, the D-term potential Eq.\eqref{eq:potD} would need to be amended to account for such deformations. The details of this was worked out in \cite{Anderson:2009nt, Anderson:2009sw}, and we do not expect it to change significantly in the torsional case. For a bundle which is close to polystability the slope may be roughly expanded as
\begin{equation*}
\mu(V)=\sum_i\mu(E_i)+ C_i^2 ,
\end{equation*}
where the $E_i$'s are sub-sheaves of $V$ and $C_i$ refers to deformations away from polystability, that is, ``matter" fields in the effective theory. This allows sub-sheaves of negative slope in general. However, away from the poly-stable locus there can be no obstructions for the hermitian moduli  from the D-terms as we saw above. These bundle configurations are hence less interesting from a moduli stabilization point of view.

\section{On the $\End(TX)$-Valued Connection}
\label{sec:EndTconn}
For completeness, we also note that we expect a similar D-term like term to appear for the $\End(TX)$-valued connection $\nabla$. Indeed, the ten-dimensional action has a term
\begin{equation}
\label{eq:DtermR}
\a\int_{M_{10}}e^{-2\phi}\,\tr\,(g^{MN}\: R_{MN})^2\:,
\end{equation}
which for a Minkowski space--time vacuum means that 
\begin{equation}
\label{eq:YM-R}
\omega\lrcorner R=0\:.
\end{equation}
Recall that $\nabla$ is also required to be holomorphic by $\delta W=0$. From \eqref{eq:YM-R} we see that $\nabla$ is also required to satisfy the Yang-Mills condition. That $\nabla$ should be Yang-Mills has also been pointed out in e.g. \cite{Ivanov:2009rh, Martelli:2010jx}. Until now, we have taken $\nabla$ to be a generic holomorphic Yang-Mills connection on $\End(TX)$. It is however known that this connection should depend on the other fields in some particular way \cite{Hull1986187, Sen1986289, 0264-9381-4-6-027, Melnikov:2012cv, Melnikov:2012nm, Hull198651, Hull1986357, Becker:2009df, Melnikov:2014ywa}. Moreover, the deformations $\kappa_2$ of this connection correspond to $\OO(\a)$ field-redefinitions of the other fields \cite{delaOssa:2014msa}.

With the usual field choice, the connection should be the Hull connection to $\OO(\a)$ in order to have a supersymmetry invariant theory \cite{Hull1986187, Sen1986289}
\begin{equation}
\label{eq:field}
\nabla=\nabla^-+\OO(\a)=\nabla^{LC}-\frac{1}{2}H+\OO(\a)\:,
\end{equation}
where the reduction in orders of $\a$ on the right hand side is due to that the connection always appears with an extra factor of $\a$. Here $\nabla^{LC}$ is the Levi-Civita connection. This connection also satisfies the supersymmetry conditions to first order, but not to higher orders \cite{delaOssa:2014msa, Ivanov:2009rh}.\footnote{It should be mentioned that it was shown in \cite{delaOssa:2014msa} that the supersymmetry conditions remain the same to $\OO(\a^2)$, provided $\nabla$ is taken to be an instanton connection at $\OO(\a)$. That is,  
\begin{equation*}
R_{mn}\gamma^{mn}\eta=0+\OO(\a^2)\:.
\end{equation*}
Note that this is not true in general for $\nabla^-$ at $\OO(\a)$ \cite{Ivanov:2009rh, delaOssa:2014msa}.}

As we are working consistently to $\OO(\a)$ in this paper, we may set $\nabla=\nabla^-$. Actually, for compact solutions it is well known that the flux $H$ is of $\OO(\a)$ \cite{Gauntlett:2001ur}, see also Appendix \ref{app:SU3}. It follows that we can choose 
\begin{equation*}
\nabla=\nabla^{LC}+\OO(\a)\:.
\end{equation*}
Note also that to the given order in $\a$ we can assume that $\nabla$ is hermitian. This follows since at zeroth order in $\a$ the internal torsion vanishes, and the Levi-Civita connection equals the Chern connection. In particular, by the theorem of Li and Yau \cite{MR915839} it follows that the bundle {\it $\End(TX)$ is stable}.

Moreover, as $R$ is type $(1,1)$ by construction at this order, it follows that $\Delta_2\in\ker(\R)$ for all $\Delta_2$. In other words,
\begin{equation*}
\ker(\R)=H^{(0,1)}(TX)\:,
\end{equation*}
and $\R$ is the zero map as part of the Atiyah map $\C_{TX}$. By a similar argument as in Theorem \ref{tm:ImH2=ImF} it also follows that $\R$ is also trivial as part of $\C_\g$. We thus see that none of the infinitesimal field redefinitions, corresponding to elements of $H^{(0,1)}(\End(TX))$, are lifted. Moreover,  the condition to be in $\ker(\H)$ becomes a condition on the true moduli $(\Delta_2,\alpha_2)$. Also, to the given order the Yang-Mills condition \eqref{eq:YM-R} is now automatically satisfied, since this is just the Ricci-flatness condition for the zeroth order K\"ahler geometry. D-terms related to this connection of the form \eqref{eq:DtermR} can therefore be ignored.

We hence see that with the standard field choice \eqref{eq:field} to the order we are working at the $\End(TX)$-moduli effectively factor out of the story. The $\R$-part of the maps $\C_{TX}$ and $\C_\g$ is trivial and cannot lift complex structure moduli or deformations of $\nabla$ respectively. As mentioned, non-trivial deformations of this connection correspond to field redefinitions away from the usual field choice \cite{delaOssa:2014msa}. For mathematical rigor, we shall continue to carry the $\End(TX)$-moduli with us for the remainder of the paper, but it should be kept in mind that these moduli are redundant from a physics point of view. Note also that there is still an $\R$-part in the definition of $\H_\Delta$ given in \eqref{eq:HDelta} which cannot be ignored a priori.

\section{Yukawa Couplings and Obstructions}  
\label{sec:yuk}
Having seen how the infinitesimal deformations work, at least up to second order in deformations of the superpotential, it is interesting to consider higher order deformations of the theory. Generically, it is known that not all infinitesimal deformations can be integrated to finite deformations. The barriers to doing so are known as \emph{obstructions} in the mathematics literature. For a holomorphic structure $\bar D$, the condition for the deformations to be unobstructed is that they are in the kernel of the obstruction map
\begin{equation*}
\kappa\::\:H^{(0,1)}(\Q)\rightarrow H^{(0,2)}(Q)\:,
\end{equation*}
often also referred to as the Kuranishi map \cite{kuranishi1971deformations, kobayashi2014differential, kodaira2004complex}. The true moduli of the theory are thus the ones in the kernel of this map, while deformations not in this kernel will be obstructed. These obstructions are known to correspond to higher order Yukawa couplings in the four-dimensional effective theory \cite{Berglund:1995yu}. To show exactly how this works requires us to do higher order deformations of the superpotential Eq.\eqref{eq:suppot}, and show how these Yukawa couplings correspond to obstructions in the deformation theory of $\bar D$. This is quite involved and we leave the full treatment  for future work, see however \cite{2015arXiv150307562G, 2015D-GF-S}. Instead we only investigate a couple of features of the obstructions here, in particular for compactifications where $X_0$ is Calabi-Yau. It should also be noted that obstructions and their correspondence to Yukawa couplings have been considered at length in the literature before, see e.
g. \cite{green2012superstring, Berglund:1995yu, polchinski1998string2, Anderson:2009ge, Anderson:2011ty}. 

In terms of holomorphic structures defined by an extension sequence, it can be shown that the obstruction maps in the corresponding long exact sequence commute with the other induced maps in cohomology \cite{huybrechts1995tangent, Anderson:2011ty, 2015D-GF-S}
\begin{align}
...\rightarrow H^0(E)&\xrightarrow{\H_0} H^{(0,1)}(T^*X)\rightarrow H^{(0,1)}(\Q)\rightarrow H^{(0,1)}(E)\notag\\
\label{eq:Obstr}
&\;\;\;\;\;\;\;\;\;\;\downarrow\:\kappa_{T^*X}\;\;\;\;\;\;\;\;\;\;\downarrow\:\kappa\;\;\;\;\;\;\;\;\;\;\;\;\downarrow\:\kappa_E\\
&\xrightarrow{\H} H^{(0,2)}(T^*X)\rightarrow H^{(0,2)}(\Q)\xrightarrow{\rho} H^{(0,2)}(E)\rightarrow0\notag\:,
\end{align}
where the last zero follows from the slope-zero stability of $T^*X$. The obstruction map $\kappa_E$ can further be sandwiched between obstruction maps of the bundle and base as
\begin{align}
0\rightarrow &H^{(0,1)}(\End(TX))\oplus H^{(0,1)}(\End(V))\rightarrow H^{(0,1)}(E)\rightarrow H^{(0,1)}(TX)\notag\\
\label{eq:Obstr1}
&\;\;\;\;\;\;\;\;\downarrow\:\kappa_{\End(TX)}\;\;\;\;\;\;\;\;\;\;\;\;\downarrow\:\kappa_{\End(V)}\;\;\;\;\;\;\;\;\;\;\downarrow\:\kappa_E\;\;\;\;\;\;\;\;\;\;\;\;\downarrow\:\kappa_{TX}\\
\xrightarrow{\F+\R} &H^{(0,2)}(\End(TX))\oplus H^{(0,2)}(\End(V))\xrightarrow{\rho_E} H^{(0,2)}(E)\rightarrow H^{(0,2)}(TX)\rightarrow...\notag\:,
\end{align}
where we have named the map $\rho_E$ as it will appear in the following computations. 

\subsection{Obstructions for Calabi-Yau Compactifications}
In order to say more about the obstructions, it is perhaps most enlightening to consider reasonably explicit examples of compactifications, where more is known about the individual obstruction maps. For example, if we assume that
\begin{equation}
\label{eq:VanishObstr}
\kappa_{TX}=\kappa_{T^*X}=\kappa_{\End(TX)}=0\:.
\end{equation}
This holds true for Calabi-Yau compactifications, where we also use the field choice described in Section \ref{sec:EndTconn} where the $\End(TX)$-connection is taken to be the Levi-Civita connection. As will be argued at the end of this section, we believe that \eqref{eq:VanishObstr} holds true for cases where $X_0$ is Calabi-Yau as well. Recall also that in this field choice we have $\R_{TX}=\R_\g=0$, leading to
\begin{align*}
\C_{TX}&=\F_{TX}\\
\C_\g&=\F_g\:.
\end{align*}
It should be noted that the computation in the next paragraph leading to Eq.\eqref{eq:kerCY} was first carried out in \cite{Anderson:2011ty}, but we repeat it here for completeness.

It follows from a diagram chase of \eqref{eq:Obstr1} that
\begin{align*}
\textrm{dim}\left(\textrm{ker}(\kappa_E)\right)&=h^{(0,1)}(E)-\textrm{dim}\left(\textrm{Im}(\kappa_E)\right)=h^{(0,1)}(E)-\textrm{dim}\left(\textrm{Im}(\rho_{\End(V)}\kappa_{\End(V)})\right)\:,
\end{align*}
where $\rho_{\End(V)}$ is $\rho_E$ restricted to $H^{(0,1)}(\End(V))$, and we denote by $h$ the dimension of the corresponding cohomologies. It follows that
\begin{equation}
\begin{aligned} \label{eq:kerCY}
\textrm{dim}\left(\textrm{ker}(\kappa_E)\right)&\ge h^{(0,1)}(E)-\textrm{dim}\left(\textrm{Im}(\rho_{\End(V)})\right)\\
&=h^{(0,1)}(E)-h^{(0,2)}(\End(V))+\textrm{dim}\left(\textrm{Im}(\F)\right)\\
&=h^{(0,1)}(\End(TX))+h^{(0,1)}(\End(V))+h^{(0,1)}(TX)-h^{(0,2)}(\End(V))\\
&=h^{(0,1)}(\End(TX))+h^{(0,1)}(TX)\:,
\end{aligned}
\end{equation}
where in the last equality we have used that $h^{(0,1)}(\End(V))=h^{(0,2)}(\End(V))$ by Serre duality. Notice that the moduli $h^{(0,1)}(\End(TX))$ corresponding to the field redefinitions are unobstructed.

Next, consider the full obstruction map $\kappa$ in \eqref{eq:Obstr} for compactifications where \eqref{eq:VanishObstr} applies, and in particular Calabi-Yau compactifications. As $\kappa_{T^*X}=0$, it follows that any nontrivial element of $\textrm{Im}(\kappa)$ must have a nontrivial image in $H^{(0,2)}(E)$ under $\rho$ corresponding to a non-trivial image of $\kappa_E$. It follows that
\begin{align*}
\textrm{dim}\left(\textrm{ker}(\kappa)\right)&=h^{(0,1)}(\Q)-\textrm{dim}\left(\textrm{Im}(\kappa)\right)\ge h^{(0,1)}(\Q)-\textrm{dim}\left(\textrm{Im}(\kappa_E)\right)\\
&=h^{(0,1)}(T^*X)-\textrm{dim}\left(\textrm{Im}(\H_0)\right)+\textrm{dim}\left(\textrm{ker}(\kappa_E)\right)-\textrm{dim}\left(\textrm{Im}(\H)\right)\\
&\ge h^{(0,1)}(T^*X)-\textrm{dim}\left(\textrm{Im}(\H_0)\right)+h^{(0,1)}(\End(TX))+h^{(0,1)}(TX)-\textrm{dim}\left(\textrm{Im}(\H)\right)\:.
\end{align*}
Recall that 
\begin{equation*}
\textrm{dim}\left(\textrm{Im}(\H)\right)\le\textrm{dim}\left(\textrm{Im}(\C_\g)\right)+\textrm{dim}\left(\textrm{Im}(\H_\Delta)\right)=\textrm{dim}\left(\textrm{Im}(\F_{TX})\right)+\textrm{dim}\left(\textrm{Im}(\H_\Delta)\right)\:,
\end{equation*}
where we have used Theorem \ref{tm:ImH2=ImF} and the fact that $\textrm{Im}(\C_{TX})=\textrm{Im}(\F_{TX})$ in this field choice. It follows that
\begin{align}
\label{eq:Limobstr}
\textrm{dim}\left(\textrm{ker}(\kappa)\right)-h^{(0,1)}(\End(TX))\ge 
h^{(0,1)}(T^*X)-\textrm{dim}\left(\textrm{Im}(\H_0)\right)+\textrm{dim}\left(\textrm{ker}(\H_\Delta)\right)\:.
\end{align}
This shows that the number of true moduli in such compactifications, given by the left hand side of \eqref{eq:Limobstr}, will always be greater or equal to the number of original K\"ahler and complex structure moduli of the base, not lifted by the procedure outlined in the previous sections. That is, the {\it massless moduli} of the base. They need not correspond to complex structure or K\"ahler moduli anymore though, as they may have been replaced by unobstructed bundle moduli. We also note that we always carry with us the moduli corresponding to field redefinitions, which are counted by $h^{(0,1)}(\End(TX))$. Modulo such field redefinitions, they can be ignored.

It should again be stressed that the above calculation only holds for cases where \eqref{eq:VanishObstr} is true. In particular for Calabi-Yau compactifications. We believe that the result also holds for more general compactifications, in particular for compactifications where $X_0$ is Calabi-Yau. Indeed, compactifications with a zeroth order compact Calabi-Yau geometry satisfy the $\p\bp$-lemma as shown Appendix \ref{app:exp}, and which is enough for $\kappa_{TX}$ to vanish \cite{tian1987smoothness}. As we also argue in Appendix \ref{app:exp}, the Hodge-diamond of $X$ for such compactifications is not affected by $\a$-corrections. In particular, the infinitesimal K\"ahler moduli space is given by 
\begin{equation*}
H^{(1,1)}_{\bp}(X)=H^{(1,1)}_{\bp}(X_0)\:.
\end{equation*}
Hence the \kahler moduli are counted by the \kahler moduli of the zeroth order base, which are unobstructed. Finally, the infinitesimal deformations of the zeroth order tangent bundle $H^{(0,1)}(\End(TX_0))$ of a Calabi-Yau are also unobstructed. These also correspond to the deformations of the connection $\nabla=\nabla^{LC}+\OO(\a)$ in the usual field choice. These arguments suggest that \eqref{eq:VanishObstr} holds true also for cases where $X_0$ is Calabi-Yau. They are however somewhat hand-wavy and remain to be confirmed with more mathematical rigor. This is a subject of further study.
 
\section{Calabi-Yau Compactifications, and Surjective and Injective Maps} 
\label{sec:examples}
Having discussed the moduli space, it is instructive to consider some examples. The first, and perhaps simplest, case to consider are Calabi-Yau compactifications, as discussed above when we considered the obstructions. One might expect that there is not much more to say for such examples then what is already known. There are however a couple of important points we wish to emphasize, related to the inclusion of the Bianchi identity in the moduli story. First note that in these examples there is of course no torsion, i.e. $\hat H=0$, but that does not mean that the $\H$-map is trivial. Indeed, as we saw with Theorem \ref{tm:ImH2=ImF}, the number of bundle moduli lifted by $\C_\g$ is the same as the number of complex structure moduli lifted by the original Atiyah map $\C_{TX}$. Moreover, even tough the torsional part of $\H_\Delta$ vanishes this does not a priori guarantee that $\H_\Delta$ is trivial, and this can lead to additional stabilization of complex structure moduli. We use the standard field choice of 
Section \ref{sec:EndTconn} in this section, so
\begin{align*}
\C_{TX}=\F_{TX}\:&:\;\;\;H^{(0,1)}(TX)\rightarrow H^{(0,2)}(\End(V))\\
\C_\g=\F_\g\;\;\;\:&:\;\;\;H^{(0,1)}(\End(V))\rightarrow H^{(0,2)}(T^*X)\:.
\end{align*}

\subsubsection*{The Standard Embedding}
Let us consider the simplest example of a Calabi-Yau compactification, namely the standard embedding. In this example, we are at the locus where we identify the gauge connection $A$ with the tangent bundle connection $\nabla$. This is often referred to as the $(2,2)$-locus from a sigma model perspective. 

At the standard embedding locus, we have $F=R$, which given the discussion in Section \ref{sec:EndTconn}, shows that $\F=\H=0$. One might think that it is only the $\F_\g$-map that vanishes in this case, but it turns out that $\H_\Delta$ is trivial in this case as well. Indeed, as the curvatures are the same, we find for a complex structure $\Delta_a$ that $\kappa_a=\alpha_a$, which gives 
\begin{equation*}
\H_\Delta(\Delta_a)=0\;\;\;\forall\;\Delta\in H^{(0,1)}(TX)\:,
\end{equation*}
and so $\H_\Delta=0$ as well. This means that the corresponding long exact sequences in cohomology split into short exact sequences, so that
\begin{equation*}
H^{(0,1)}(\Q)=H^{(0,1)}(T^*X)\oplus H^{(0,1)}(E)=H^{(0,1)}(T^*X)\oplus H^{(0,1)}(TX)\oplus H^{(0,1)}(\g)\:.
\end{equation*}
Furthermore, the obstructions are given by $\rm{Im}(\kappa)$, whose dimension in this case is given by 
\begin{equation*}
\textrm{dim}\left(\textrm{Im}(\kappa)\right)=\textrm{dim}\left(\textrm{Im}(\kappa_E)\right)=\textrm{dim}\left(\textrm{Im}(\kappa_{\End(V)})\right)\:.
\end{equation*}
In other words, the moduli of the base remain unobstructed, while there are potential obstructions to the bundle moduli. This case has been examined in great detail the literature before, and the reader is referred to e.g. \cite{green2012superstring, Berglund:1995yu, polchinski1998string2} and references therein for more details. 

\subsubsection*{A Surjective Atiyah Map}
Next, we consider cases where the Atiyah map $\F_{TX}$ is surjective. Restricting ourselves to the physical bundle moduli, by Theorem \ref{tm:ImH2=ImF} the map $\F_\g$ then becomes injective. Note that this also holds outside of the Calabi-Yau locus. We thus see that in such cases, all bundle moduli are lifted. Examples of this kind with a surjective $\F_{TX}$ can be found in e.g. \cite{Anderson:2010mh}. Indeed, in Section 5 of \cite{Anderson:2010mh} an explicit example of a bundle $V$ is constructed, where the number of lifted complex structure moduli equals the dimension of $H^{(0,2)}(\End(V))$, implying that $\F_{TX}$ is surjective. Including the Bianchi identity, it follows that $\H_2$ is injective, which means that all bundle moduli corresponding to $\End(V)$ are lifted in this case. 

It should however be noted that this need not have an effect on the physical matter spectrum. Indeed, the authors of \cite{Anderson:2010mh} suggest to use the hidden $E_8$-bundle to stabilize complex structure moduli in more phenomenology oriented models. In this case one only lifts bundle moduli corresponding to deformations of the hidden bundle, and hence the physical spectrum important for phenomenology is unaffected.

\subsubsection*{Stabilising Complex Structure Moduli and an Injective Atiyah Map}
For completeness, we also consider what happens in examples where the Atiyah map $\F_{TX}$ is injective. It is clear that in such examples all complex structure moduli will get lifted, and they are hence of great phenomenological value. See e.g. \cite{Anderson:2011cza, Anderson:2011ty} for examples and developments on this front. It should also be noted that, given Theorem \ref{tm:ImH2=ImF}, one can equivalently look for example where the map $\F_\g$ is surjective. Moreover, the map $\H_\Delta$ can potentially lift additional complex structure moduli should the Atiyah map not be injective.

\section{Discussion and Conclusions}

In this paper we have studied the infinitesimal moduli space of heterotic string compactifications from the perspective of the 4d theory. By considering F-terms we have shown that some `would-be' moduli are actually heavy and removed from the effective theory. A further restriction of the moduli space occurs from D-terms, and combined with the F-term conditions we have shown that our results are equivalent to those previously obtained from the 10d theory. We have also reviewed how this may be phrased in terms of maps on coholomogies and continued the study of these maps. Finally, we have briefly considered Yukawa couplings in such theories, and considered some examples. 

Of course, an enormous amount of work remains before such torsional compactifications are fully understood and potentially able to lead to fully realistic low energy phenomenology. An obvious omission is our present lack of knowledge of the \kahler potential, although this is the subject of current work \cite{delaOssa}. It may be hoped that, given the holomorphic structures discussed, the \kahler potential will take a fairly simple and elegant form. Indeed, holomorphic structures usually come equipped with some form of Weil-Peterson metric on their moduli space, and one can speculate that the \kahler metric one obtains upon dimensional reduction corresponds to such a metric. However, for this present time this remains an open question. The story with Yukawa couplings is also far from complete. In particular, the connection between higher order deformations of the superpotential and obstructions has not yet been made explicit and it would be interesting to see how the details of this emerge. Knowledge of the Yukawa couplings is also very important for phenomenological purposes as well. 

It would also be very interesting to study explicit examples of compactifications with torsion. Compactifications where a large volume Calabi-Yau locus exists are fairly easy to construct once the zeroth order \kahler geometry is known, and it would be interesting to investigate further what effects (if any) the $\a$-generated torsion has for lifting further moduli. Studying examples where no zeroth order limit exist is more challenging. Examples of this kind found in the literature \cite{Dasgupta:1999ss, Becker:2006et} have been shown to negate some of the assumptions we make. In particular $H^{(0,1)}(X)\neq0$ \cite{Cyrier:2006pp}. It is hence less clear if $H^{(0,1)}(\Q)$ counts the true moduli for these types of compactifications, but it can be taken as a conjecture. See also \cite{GarciaFernandez:2015hja} for more details on this. In the longer term, it may even be hoped that there is the possibility of constructing examples with all moduli either removed from the low energy theory, or otherwise 
stabilised with phenomenologically acceptable masses. Investigations of other aspects of the low energy phenomenology, e.g. the number of Standard Model generations, exotic matter present, etc. may also be possible and interesting.

\section*{Acknowledgements}
We are grateful to  Lara Anderson,  Mario Garc\'\i a-Fern\'andez, James Gray, Dan Israel, Spiro Karigiannis, Magdalena Larfors, Andre Lukas and Daniel Plencner for useful discussions.  EES is supported by the ILP LABEX (under reference ANR-10-LABX-63), and by French state funds managed by the ANR within the Investissements dAvenir program under reference ANR-11-IDEX-0004-02.  XD's research is supported in part by the EPSRC grant BKRWDM00.

\appendix

\section{Compactifications and $SU(3)$-structure}
\label{app:SU3}
In this appendix, we review some results concerning $SU(3)$-structure manifolds and string compactifications. The ten-dimensional geometry is assumed to have the form of a warped direct product,
\begin{equation*}
M_{10}=M_4\times X_6,
\end{equation*}
where $M_4$ is four-dimensional space-time, and $X_6$ is the compact internal space. We will use small roman indices $\{m,n,..\}$ to denote real indices on $X_6$, and greek indices to denote indices on $M_4$ whenever needed. The ten-dimensional supercharge given by the Majorana-Weil spinor $\epsilon$ decomposes as
\begin{equation*}
\epsilon=\rho\otimes\eta,
\end{equation*}
where $\rho$ is a four-dimensional space-time spinor and $\eta$ is a spinor on $X_6$. The spinor $\eta$ defines an $SU(3)$-structure, with  the non-degenerate two-form $\omega$ and nowhere vanishing well defined three-form $\Psi$ on $X_6$ given by
\begin{align*}
\omega_{mn}&=-i\eta^\dagger\gamma_{mn}\eta,\\
\Psi_{mnp}&=\eta^T\gamma_{mnp}\eta,
\end{align*}
where $\gamma_m$ are the gamma matrices that satisfy the Clifford algebra in six dimensions, and $\gamma_{m_1m_2\ldots m_p}$ denotes the totally antisymmetric product of $p$ gamma matrices.  
These forms satisfy the usual $SU(3)$-structure identities
\begin{equation*}
\omega\wedge\Psi=0,\;\;\;\;\frac{i}{\vert\vert\Psi\vert\vert^2}\Psi\wedge\bar\Psi=\frac{1}{6}\omega\wedge\omega\wedge\omega.
\end{equation*}
The three-form $\Psi$ defines an almost complex structure $J$ on $X_6$ by
\begin{equation}
\label{eq:complexstr}
{J_m}^n=\frac{{I_m}^n}{\sqrt{-\frac{1}{6}\tr I^2}},
\end{equation}
where the tangent bundle endomorphism $I$ is given by
\begin{equation*}
{I_m}^n=(\textrm{Re}\Psi)_{mpq}(\textrm{Re}\Psi)_{rst}\epsilon^{npqrst}.
\end{equation*}
The normalization in \eqref{eq:complexstr} is needed so that $J^2=-1$. Note also that the complex structure $J$ is independent of rescalings of $\Psi$.

\subsection*{$SU(3)$-Structures}
A general $SU(3)$-structure has five torsion classes, $(W_0, W_1^\omega,W_1^\Psi,W_2,W_3)$, where \cite{0444.53032, 1024.53018, LopesCardoso:2002hd, Gauntlett:2003cy}
\begin{align*}
\d\omega&=-\frac{12}{\vert\vert\Psi\vert\vert^2}\textrm{Im}(W_0\bar\Psi)+W_1^\omega\wedge\omega+W_3\\
\d\Psi&=W_0\:\omega\wedge\omega+W_2\wedge\omega+\bar W_1^\Psi\wedge\Psi\:.
\end{align*}
Here $W_0$ is a complex function, $W_2$ is a primitive $(1,1)$-form,  $W_3$ is  a real  primitive three-form of type $(1,2)+(2,1)$, $W_1^\omega$ is a real one-form, and $W_1^\Psi$ is a $(1,0)$-form. The one forms $W_1^\omega$ and $W_1^\Psi$ are known as the Lee-forms of $\omega$ and $\Psi$ respectively, and they are given by
\begin{align*}
W_1^\omega&=\frac{1}{2}\omega\lrcorner\d\omega\\
W_1^\Psi&=\frac{1}{\vert\vert\Psi\vert\vert^2}\Psi\lrcorner\d\bar\Psi.
\end{align*}
It should be noted that $W_2=W_0=0$ is equivalent to the vanishing of the Nijenhaus tensor, and therefore equivalent to $X_6$ being complex. Note also that under a rescaling $\Psi\rightarrow\lambda\Psi$ which leaves the complex structure invariant, $\lambda\in\mathbb{C}^*$, the lee-forms and $W_3$ remain invariant, while 
\begin{equation*}
W_0\rightarrow\lambda W_0,\;\;\;\;W_2\rightarrow\lambda W_2.
\end{equation*}
Interestingly, it is only the torsion classes which spoil integrability of the complex structure that scale with $\lambda$. 

\subsection*{The Strominger System}
Let us recall the six-dimensional equations that should be solved in order to have a supersymmetric solution to the equations of motion up to first order in $\a$,
\begin{align}
\label{eq:apstrom1}
\d(e^{-2\phi}\Psi)=\d\Omega&=0+\OO(\a^2)\\
\label{eq:apstrom2}
\d(e^{-2\phi}\omega\wedge\omega)&=0+\OO(\a^2)\\
\label{eq:apstrom3}
e^{2\phi}*\d(e^{-2\phi}\omega)=i(\p-\bp)\omega&=H+\OO(\a^2)\\
\label{eq:apstrom4}
\omega\lrcorner F= 0 +\OO(\a)~, \qquad \Omega\wedge F&=0+\OO(\a)\\
\label{eq:apstrom5}
\omega\lrcorner R=0+\OO(\a)~,\qquad\Omega\wedge R&=0+\OO(\a)\:,
\end{align}
where $\Omega=e^{-2\phi}\Psi$. Note the reduction of the order of $\a$ in the last two equations due to the fact that these curvature terms only appear at first order. in the geometry. This set of equations together with the Bianchi Identity
\begin{equation*}
\d H=\frac{\a}{4}(\tr\: F\wedge F-\tr\: R\wedge R)+\OO(\a^2) ,
\end{equation*}
is commonly referred to as the Strominger system. If $X$ is compact, then a simple application of Stokes theorem shows that the flux $H$ is of order $\a$ or smaller \cite{Gauntlett:2003cy}
\begin{equation*}
\vert\vert e^{-\phi}H\vert\vert^2=\int_X e^{-2\phi}H\wedge*H=-\int_X H\wedge \d(e^{-2\phi}\omega)=\OO(\a)\:,
\end{equation*}
as $\d H=\OO(\a)$. If we exclude non-integer orders of $\a$, it follows that $H=\OO(\a)$, and also $\d\phi=\OO(\a)$ leaving us with $\d\omega=\OO(\a)$ and a Calabi-Yau geometry at zeroth order. 

\section{The $\a$-Expansion and the Strominger System}
\label{app:exp}
In this appendix, we comment on the $\alpha'$-expansion of heterotic supergravity. It should be mentioned that studies of heterotic supergravity in terms of the $\a$-expansion have been carried out before, see in particular \cite{Witten:1986kg} and e.g. \cite{Gillard:2003jh, Anguelova:2010ed}. Moreover, it should be stressed that the types of geometries discussed here are only a subset of the geometries for which the main results of the paper applies. In particular, the geometries discussed in this appendix all have a large volume, $\a\rightarrow0$ Calabi-Yau limit which is smooth and compact.

Note first that the $\a$-expansion is really an expansion in terms of the dimensionless parameter 
\begin{equation*}
\beta=\alpha'\big/[\textrm{Vol}(X_6)]^{\frac{1}{3}}\:,
\end{equation*}
where we assume a large volume, and thus that the supergravity limit is valid, but we'll mostly suppress this in the following. We also stress the philosophy adopted in this expansion, that $\a$-corrections are small. In particular, we are assuming a well-defined topology zeroth order in $\a$, and that the $\a$-corrections can only change the geometry which we assume to be smooth at each order in $\a$. Corrections can then be computed order by order using the lower order geometry. This assumption excludes examples of the kind found by Dasgupta et al \cite{Dasgupta:1999ss}, later expanded upon by Yau et. al. \cite{Becker:2006et}, as these geometries are highly singular in the $\a\rightarrow0$ limit. It should however be stressed that within these solutions it is hard to find a large-volume limit and the supergravity approximation is perhaps less applicable in this sense \cite{Becker:2003gq, Melnikov:2014ywa}. Hence, the $\a$-expansion scheme does not cover all solutions to heterotic supergravity, only the 
large volume regime. 

The various fields in the theory have an $\alpha'$ expansion. For a generic field $\Phi$ we have
\begin{equation*}
\Phi=\Phi_0+\a\Phi_1+\cdots~.
\end{equation*} 
The zeroth order geometry when $\a\rightarrow0$ is denoted by $X_0$. We assume that this geometry is smooth and compact, and as we will see below this implies that the geometry is Calabi-Yau. We also mention that we do not expand the complex structure. That is, a complex structure can be defined on $X$ without any reference to a metric, and thus no inherent size associated to it so that
\begin{equation*}
J=J_0\:.
\end{equation*}
This simplifies our lives, as it means the holomorphic structures do not need an $\a$-expansion. In particular, $\partial$ and $\bar\partial$ are not have $\a$ corrections. In addition, clearly if a form is of type $(n,m)$ at zeroth order, it remains so at higher orders. That is, the Hodge-type is preserved by the expansion. We also note that since the complex structure is size-independent, we expect the geometry to remain complex at higher orders in $\a$.  A complex geometry is equivalent to
\begin{equation*}
\bp^2=0,
\end{equation*}
and as this equation receives no corrections, it will remain true to all orders. From this, we conjecture that higher perturbative $\a$-corrections will respect a complex base. Of course, this can be spoiled by non-perturbative corrections, but we do not consider these in the current paper.

Finally, we collect some useful facts about geometries where the large volume limit is a compact Calabi-Yau. In particular, recall that a sufficient condition for a complex manifold $(X,J)$ to satisfy the $\p\bp$-lemma is the existence of a K\"ahler form compatible with $J$. Since the complex structure does not change under $\a$-corrections, and since there must exist a K\"ahler form $\omega_0$ corresponding to the zeroth order Calabi-Yau geometry $X_0$, it follows that the corrected geometry $X$ satisfies the $\p\bp$-lemma. Moreover, as the Dolbeault operator $\bp$ remains unchanged under $\a$-corrections, we can conclude that the Hodge-diamond of $X$ does not change either. Indeed, as the Dolbeault cohomologies of a Calabi-Yau manifold are topological, and as we have seen $X$ admits a K\"ahler metric, any change to this at higher orders in $\a$ implies topological changes of $X$ which contradicts the assumptions of the $\a$-expansion. Note that a similar statement need not hold for bundle valued cohomologies, as the connections on the given bundles can potentially receive corrections, and the bundles need no longer be holomorphic in general.


\providecommand{\href}[2]{#2}\begingroup\raggedright\endgroup

\end{document}